\documentclass[preprint]{imsart}

\usepackage{amsthm,amsmath,amssymb,natbib}
\RequirePackage[colorlinks,citecolor=blue,urlcolor=blue]{hyperref}
\usepackage[plain,noend]{algorithm2e}
\usepackage{subfig}
\usepackage{graphicx}

\startlocaldefs
\makeatletter
\renewcommand*\env@matrix[1][*\c@MaxMatrixCols c]{%
  \hskip -\arraycolsep
  \let\@ifnextchar\new@ifnextchar
  \array{#1}}
\makeatother
\newtheorem{defn}{Definition}[section]
\newtheorem{lem}{Lemma}[section]
\newtheorem{thm}{Theorem}[section]
\newtheorem{prop}{Proposition}[section]

\newcommand{\Tr}{\text{Tr}}
\DeclareMathOperator*{\argmin}{arg\,min}
\DeclareMathOperator*{\argmax}{arg\,max}
\endlocaldefs

\begin{document}

\begin{frontmatter}

\title{Selective inference in regression models with groups of variables}
\runtitle{Selective inference with groups of variables}

\begin{aug}
\author{\fnms{Joshua} \snm{Loftus}\corref{}\ead[label=e1]{joftius@stanford.edu}\thanksref{t1,t2}}
\and
\author{\fnms{Jonathan} \snm{Taylor}\ead[label=e2]{jonathan.taylor@stanford.edu}\thanksref{t2}}
\thankstext{t1}{Supported in part by NIH grant 5T32GMD96982.}
\thankstext{t2}{Supported in part by NSF grant DMS 1208857 and AFOSR grant 113039.}

\runauthor{J. R. Loftus \and J. E. Taylor}

\affiliation{Stanford University}
\address{Department of Statistics\\Sequoia Hall\\390 Serra
  Mall\\Stanford, CA\\
\printead{e1,e2}}
\end{aug}

\begin{abstract}
We discuss a general mathematical framework for selective inference
with supervised model selection procedures characterized by 
quadratic forms in the outcome variable. Forward stepwise with groups
of variables is an important special case as it allows models with
categorical variables or factors.
Models can be chosen by AIC, BIC, or a fixed number of steps.
We provide an exact significance test for
each group of variables in the selected model based on
appropriately truncated $\chi$ or $F$ distributions for the cases of
known and unknown $\sigma^2$ respectively.
An efficient software
implementation is available as a package in the R statistical
programming language.
\end{abstract}



\end{frontmatter}

\section{Introduction}

A common strategy in data analysis 
is to assume a class of models $\mathcal M$
specified {\em a priori} contains a particular model $M \in \mathcal
M$ which is well-suited to the data. Unfortunately, when the data has
been used to pick a particular model $\hat M$, inference about
the selected model is usually invalid. For example, in the regression
setting with many predictors if the ``best'' predictors are chosen by
lasso or forward stepwise then the classical $t$, $\chi^2$, or $F$
tests for their coefficients will be anti-conservative. The most
widespread solution is data splitting, where one subset of the data is
used only for model selection and another subset is used only for
inference. This results in a loss of accuracy for model selection and
power for inference.

Recently, \cite{lockhart2014, lee2013exact} developed methods of
adjusting the classical significance tests to account for model
selection. These methods make use of the fact that the model chosen by
lasso and forward stepwise can be characterized by affine inequalities
in the data. Inference conditional on the selected model requires
truncating null distributions to the affine selection region. The
present work provides a new and more general mathematical framework
based on quadratic inequalities.
Forward stepwise with groups
of variables serves as the main example which we develop in some
detail.

Allowing variables to be grouped provides important modeling
flexibility to respect known structure in the predictors.
For example, to fit hierarchical interaction models as in
\cite{lim2013glinternet}, non-linear models with groups of spline
bases as in \cite{alex2015gamsel}, or factor models where groups
are determined by scientific knowledge such as variables in
common biological pathways.
Further, in the presence of
categorical variables, without grouping a model selection procedure
may include only some subset of the levels that variable. 

All methods described here are implemented in the
\texttt{selectiveInference} R package \citep{SIR} available on CRAN.

\subsection{California county health data example}

Information on health and various other indicators has been aggregated
to the county level by \cite{countydata}. We attempt to fit a linear
model for California counties
with outcome given by log-years of potential life lost, and
search among 31 predictors using forward stepwise. This example has
groups only of size 1, but still serves to illustrate the procedure.
With the BIC
criteria for choosing model size the resulting model has 8 variables
listed in Table~\ref{tab:counties}. The naive $p$-values computed from
standard hypothesis tests for linear regression models are biased to
be small because we chose the 8 best predictors out of 31. 

\begin{table}[ht]
\caption{{Significance test $p$-values for predictors chosen
by forward stepwise with BIC.}}
\label{tab:counties}
\centering
\begin{tabular}{r|r|rr}
  \hline
Step & Variable & Naive & Selective \\ 
  \hline
1 & \%80th Percentile Income & $<$0.001 & 0.001 \\ 
2 &   Injury Death Rate & $<$0.001 & 0.086 \\ 
3 &   Chlamydia Rate & 0.078 & 0.287 \\ 
4 &   \%Obese & $<$0.001 & 0.170 \\ 
5 &   \%Receiving HbA1c & $<$0.001 & 0.335 \\ 
6 &   \%Some College & 0.005 & 0.864 \\ 
7 &   Teen Birth Rate & 0.071 & 0.940 \\ 
8 &   Violent Crime Rate & 0.067 & 0.179 \\ 
   \hline
\end{tabular}
\end{table}

The selective $p$-values reported here are the $T\chi$ statistics
described in Section~\ref{sec:Tchi}. These have been adjusted to
account for model selection. Using the data to select a model and then
conducting inference about that model means we are randomly choosing
which hypotheses to test. In this scenario, \cite{fithian2014optimal}
propose controlling the {\em selective type 1 error} defined as
\begin{equation}
  \label{eq:type1err}
  \mathbb P_{M,H_0(M)}(\text{reject } H_0|M \text{ selected}).
\end{equation}
The notation $H_0(M)$ emphasizes the fact that the hypothesis depends
upon the model $M$. Scientific studies reporting $p$-values which
control \eqref{eq:type1err} would not suffer from model selection bias,
and consequently would have greater replicability.

One of the most common methods which controls selective type 1 error
is data splitting, where independent subsets of the data are used for
model selection and inference separately. Unfortunately, by using less
data to select the model and less data to conduct tests this method
suffers from loss of model selection accuracy and lower power. The
$T\chi$ and $T_F$ hypothesis tests described in the present work
control \eqref{eq:type1err} while using the whole data for both
selection and inference.

\subsection{Background: the affine framework}

Many model selection procedures, including lasso, forward stepwise,
least angle regression, marginal screening, and others, can be
characterized by affine inequalities. In other words, for each $m \in
\mathcal M$ there exists a matrix $A_m$ and vector $b_m$ such that
$\hat M(y) = m$ is equivalent to $A_my \leq b_m$. Hence, assuming 
\begin{equation}
  \label{eq:normalmodel}
  y = \mu + \epsilon, \quad \epsilon \sim N(0, \sigma^2 I)
\end{equation}
we can carry out inference conditional on $\hat M(y) = m$ by
constraining the multivariate normal distribution to the polyhedral
region $\{ z : A_mz \leq b_m \}$. 

One limitation of these methods is that the definition of a model $m$
includes not only the active set $E$ of variables chosen by the lasso
or forward stepwise, but also their signs. Without including signs,
the model selection region becomes a union of $2^{|E|}$
polytopes. Details are provided in \cite{lee2013exact}. The forward
stepwise and least angle regression cases \textit{without} groups of
variables appear in \cite{tibshirani2014exact}. 

In an earlier work, \cite{loftus2014significance} iterated the global
null hypothesis test of \cite{taylor2013tests} at each step of forward
stepwise with groups. The test there is only exact at the first step
where it coincides with the $T\chi$ test described here, and
empirically it was increasingly conservative at later steps.
The present work does not pursue sequential testing, and instead
conducts exact tests for all variables included in the final model.

\subsection{The quadratic framework}

Our approach will exploit the structure of model selection procedures
characterized by quadratic inequalities in the following sense. 

\begin{defn}
\label{def:quadratic}
A quadratic model selection procedure is a map $M : \mathcal X \to
\mathcal M$ determining a model $m = M(y)$, such that for any $m \in
\mathcal M$ 
\begin{equation}
\begin{aligned}
  \label{eq:quaddecomp}
  E_m &:= \{ y : M(y) = m \} = \bigcap_{j \in J_m} E_{m,j},\\
  E_{m,j} &:= \{ y : y^TQ_{m,j}y + a_{m,j}^Ty + b_{m,j} \geq 0 \}
\end{aligned}
\end{equation}
\end{defn}
The power of this definition lies in the ease with which we can
compute one-dimensional slices through $E_m$. This enables us to
compute closed-form results for certain one-dimensional
statistics and to potentially implement hit-and-run sampling
schemes. Consider tests based on $T^2 := \| Py \|^2_2$ for some fixed
projection matrix $P$. Write $u = Py/T$, $z = y - Py$, and
substitute $uT + z$ for $y$ in the definition
\eqref{eq:quaddecomp}. Conditioning on $z$ and $u$, the only remaining
variation is in $T$, and the equations above are univariate
quadratics. This allows us to determine the truncation region for $T$
induced by the model selection $M(y) = m$. We now apply this approach
and work out the details for several examples. The example in
Section~\ref{sec:TF} is a slight extension beyond the quadratic
framework that illustrates how this approach may be further
generalized. 

\section{Forward stepwise with groups}
\label{sec:groupfs}

This section establishes notation and describes the particular variant
of grouped forward stepwise referred to throughout this paper.
Readers familiar with the statistical programming language R may know
of this algorithm as the one implemented in the \texttt{step}
function. Because we are interested in groups of variables, the
stepwise algorithm is not stated in terms of univariate correlations
the way forward stepwise is often presented.

We observe an outcome $y \sim N(\mu, \sigma^2I)$ and wish to
model $\mu$ using a sparse linear model $\mu \approx X\beta$ for a
matrix of predictors $X$ and sparse coefficient vector $\beta$. We
assume the matrix $X$ is subdivided {\em a priori} into groups 
\begin{equation}
  \label{eq:groupdesign}
  X = \begin{bmatrix} X_1 & X_2 & \cdots & X_G \end{bmatrix}
\end{equation}
with $X_g$ denoting the submatrix of $X$ with all columns of the group
$g$ for $1 \leq g \leq G$.
Forward
stepwise, described in Algorithm~\ref{algo:fs}, picks at each step the
group minimizing a penalized residual sums of squares criterion.
The penalty is an AIC criterion penalizing group size, since otherwise
the algorithm will be biased toward selecting larger groups.
First we describe the case where $\sigma^2$ is known.
With a parameter $k \geq 0$, the penalized RSS criterion is 
\begin{equation}
g_1 = \argmin  \| (I - P^1_g)y \|_2^2 + k \sigma^2 \Tr(P^1_g)
\end{equation}
where $P^1_g := X_g X_g^\dagger$ and $\Tr(P^1_g)$ is the degrees of
freedom used by $X_g$. Minimizing this with respect to $g$ is
equivalent to maximizing 
\begin{equation}
  \label{eq:groupfsRSS}
RSS_1(g) := y^T P^1_g y - k \sigma^2 \Tr(P^1_g)
\end{equation}
Note that the event $E_1 := \{ z : \argmax RSS_1(g) = g_1 \}$ can be
decomposed as an intersection 
\begin{equation}
E_1 = \bigcap_{g \neq g_1} \{ z : z^T Q^1_g z + k^1_g \geq 0 \}
\end{equation}
where $Q^1_g := P^1_{g_1} - P^1_g$ and $k^1_g :=
k \sigma^2 \Tr(P^1_{g_1}-P^1_g)$. We ignore ties and use $>$ and $\geq$
interchangeably. Ties do not occur with probability 1 under fairly
general conditions on the design matrix.  

\begin{figure}
\begin{algorithm}[H]
  \caption{Forward stepwise with groups, $\sigma^2$ known}
  \label{algo:fs}
    \SetAlgoLined
    \DontPrintSemicolon
    \KwData{An $n$ vector $y$ and $n \times p$ matrix $X$ with $G$ groups, complexity penalty parameter $k \geq 0$}
    \KwResult{Ordered active set $A$ of groups included in the model at each step}
    \Begin{
    $A \leftarrow \emptyset$, $A^c \leftarrow \{ 1, \ldots, G\}$, $r_0 \leftarrow y$\;
    \For{$s=1$ \KwTo $steps$}{
        $P_{g} \leftarrow X_{g}X_{g}^\dagger$\;
        $g^* \leftarrow \argmax_{g \in A^c} \{ r_{s-1}^T P_g r_{s-1} - k\sigma^2\Tr (P_g) \}$\;

        $A \leftarrow A \cup \{ g^* \}$, $A^c \leftarrow A^c \backslash \{ g^* \}$\;
        \For{$h \in A^c$}{
            $X_h \leftarrow (I-P_{g^*}) X_h$\;
        }
        $r_s \leftarrow (I-P_{g^*}) r_{s-1}$\;
    }
    \KwRet{$A$}
    }
\end{algorithm}
\end{figure}

After adding a group of variables to the active set, we orthogonalize
the outcome and the remaining groups with respect to the added
group. So, at step $s > 1$, we have $A_s = \{ g_1, g_2, \ldots,
g_{s-1} \}$ and define 
\begin{equation}
\begin{aligned}
P_s &:= (I-P^{s-1}_{g_{s-1}})(I-P^{s-2}_{g_{s-2}}) \cdots (I-P^1_{g_1}) \\
r_s &:=  P_sy, \quad X_g^s := P_s X_g, \quad P_g^s := X_g^s (X_g^s)^\dagger
\end{aligned}
\end{equation}
Then the group added at step $s$ is
\begin{equation}
g_s = \argmax_{g \in A^c} r_s^T P^s_g r_s - k\sigma^2\Tr P^s_g
\end{equation}
Now for
\begin{equation}
E_s := \{ z \in E_{s-1} : \argmax RSS_s(g) = g_s \}
\end{equation}
we have
\begin{equation}
  E_s = E_{s-1} \cap \bigcap_{\substack{g \in A_s^c \\ g \neq g_s}} \{
  z : z^TP_s Q^s_g P_sz + k^s_g \geq 0 \} \\ 
\end{equation}
where $Q^s_g := P^s_{g_s} - P^s_g$ and $k^s_g :=
\sigma^2 k \Tr(P^s_{g_s}-P^s_g)$. We have established the following lemma. 
\begin{lem}
  \label{lem:groupfsquad}
  The selection event $E$ that forward stepwise selects the ordered
  active set $A_s = \{ g_1,
  g_2, \ldots, g_s \}$ can be written as an intersection of quadratic
  inequalities all of the form $y^T Q y + b \geq 0$, which satisfies
  Definition~\ref{def:quadratic}.
\end{lem}
We leverage this fact to compute truncation intervals for selective
significance tests based on one-dimensional slices through $E$. When
$\sigma$ is known, the truncated $\chi$ test statistic $T\chi$
described in Section~\ref{sec:Tchi} is used to test the significance
of each group in the active set. For unknown $\sigma$, the
corresponding truncated $F$ statistic $T_F$ is detailed in
Section~\ref{sec:TF}. The unknown $\sigma^2$ case alters the above
algorithm by changing the RSS criteria to the form
\begin{equation}
g_s = \argmax (r_s^Tr_s-r_s^TP_{g_s}^sr_s) \exp(-k \Tr(P_{g_s}^s))
\end{equation}
This merely replaces the additive constant with a multiplicative
one. Finally, we note that $k = 2$ corresponds to the classic AIC
criterion \cite{akaike1973information},
$k = \log(n)$ yields BIC \cite{schwarz1978estimating},
and $k = 2\log(p)$ gives the RIC criterion of
\cite{foster1994risk}. The extension of Algorithm~\ref{algo:fs} to
allow stopping using the AIC criterion is given later in
Section~\ref{sec:aic}. 

\subsection{Linear model and null hypothesis}
Once forward stepwise has terminated yielding an active set
$A$, we wish to conduct inference about the model regressing $y$ on
$X_A$ where the subscript denotes all columns of $X$ corresponding to
groups in $A$.
Throughout this paper we do not assume the
linear model is correctly specified, i.e. it is possible that $\mu
\neq X_A\beta$ for all $\beta$.
This is a strength for robustness purposes, however it also
results in lower power against alternatives where the linear model is
correctly specified and forward stepwise captures the true active
set.

Even when the linear model is not correctly specified, there is still
a well-defined best linear approximation
\begin{equation}
  \label{eq:beta}
  \beta_A^* := \argmin \mathbb E[ \|y - X_A\beta\|_2^2]
\end{equation}
with the usual estimate given by the least squares fit $\hat \beta_A =
X_A^\dagger y$.

The probability modeling assumption throughout this paper is
\eqref{eq:normalmodel}, and the null hypothesis for a group $g$ in $A$
is given by
\begin{equation}
  \label{eq:null}
  H_0(A,g) : \beta_{A,g}^* = 0 
\end{equation}
where $\beta_{A,g}$ denotes the coordinates of coefficients for group
$g$ in the linear model determined by $X_A$. Note that this null
hypothesis depends on $A$, as the coefficients for the same group $g$
in a different active set $A' \neq A$ with $g \in A'$
have a different meaning, namely
they are regression coefficients controlling for the variables in $A'$
rather than in $A$.
Several equivalent formulations of \eqref{eq:null} are
\begin{equation}
\begin{aligned}
  \label{eq:null2}
  H_0(A,g) &: X_g^T(I-P_{A\slash g})\mu = 0, \\
  H_0(A,g) &: P_g(I-P_{A \slash g})\mu = 0, \\
  H_0(A,g) &: \tilde P_g \mu = 0
\end{aligned}
\end{equation}
where $\tilde P_g$ is the projection onto the column space of
$(I-P_{A \slash g})X_g$.
We use these forms to emphasize that the probability
model for $y$ is determined by $\mu$ and not by $\beta^*$.

\section{Truncated $\chi$ significance test}
\label{sec:Tchi}

We now describe our significance test for a group of variables $g$ in
the active set $A$, under the probability model $y \sim N(\mu,
\sigma^2I)$ where we assume $\sigma$ is known or can be estimated
independently from $y$. Section~\ref{sec:TF} concerns the case where
$\sigma$ is unknown.

\subsection{Test statistic, null hypothesis, and distribution}

First we consider forward stepwise with a fixed, deterministic number
of steps $S$. Choosing $S$ with an AIC-type criterion requires some
additional work described in Section~\ref{sec:aic}.
Let $A$ denote the ordered active set and $E$ the event that $A =
\{ g_1, g_2, \ldots, g_S \}$. With $\tilde X_g = (I-P_{A\slash g})X_g$,
let $\tilde X_g = U_gD_gV_g^T$ be the singular value decomposition of
$\tilde X_g$. Note that $\tilde P_g = U_gU_g^T$.
In this notation, \textit{without model selection}, the test statistic
and null distribution
\begin{equation}
  \label{eq:chisaturatednull}
  T^2 := \sigma^{-2} \| U_g^T y \|_2^2 \sim \chi^2_{\Tr(\tilde P_g)} \text{ under }
  H_0(A,g) 
\end{equation}
form the usual regression model hypothesis test for the group $g$ of
variables in model $X_A$ when $\sigma^2$ is known.

Let $u \propto \tilde P_gy$ be a unit vector, so $y = z + Tu$ where $z =
(I-\tilde P_g)y$. Under the null, $(T,z,u)$ are independent because
$\tilde P_gy$ and $z$ are orthogonal and $(T, u)$ are independent by
Basu's theorem. We condition on $z$ and $u$ without changing the test
under the null, but it should be noted that this may result in some
loss of 
power under certain alternatives. Now the only remaining variation is
in $T$, and we still have $T^2 \sim \chi^2_{\Tr(\tilde P_g)}$ if there
is no model selection. Finally, to obtain the selective hypothesis
test we apply the following theorem with $P_A = \tilde P_g$. 
\begin{thm}
  \label{thm:tchi}
  If $y \sim N(\mu, \sigma^2I)$ with $\sigma^2$ known, and $P_A\mu = 0$
  for some $P_A$ which is constant on $\{ M(y) = A \}$, 
  then defining $r := \Tr(P_A)$, $R := P_Ay$, $u := R/\| R\|_2$, $z: = y -
  R$, the truncation interval  $M_A := \{ t \geq 0 : M(ut\sigma+z) = A
  \}$, and the observed statistic $T = \|R\|_2/\sigma$, we have 
  \begin{equation}
    \label{eq:tchilaw}
    T | (A, z, u) \sim \left.\chi_r \right|_{M_A}
  \end{equation}
  where the vertical bar denotes truncation. Hence, with $f_r$ the pdf
  of a central $\chi_r$ random variable 
  \begin{equation}
    \label{eq:tchisurv}
    T\chi := 
    \frac{\int_{M_A \cap [T, \infty]} f_r(t)dt }{\int_{M_A}
      f_r(t)dt} \sim U[0,1]
  \end{equation}
  is a $p$-value. 
\end{thm}
\begin{proof}
  First consider $P_A = P$ as fixed independently of $y$,
  and write $T_P, z_P, u_P$, and $M_A(P)$ to emphasize
  these are determined by $P$. Without selection, we know $(T_P, z_P,
  u_P)$ are independent and $T_P | (z_P, u_P) \sim \chi_r$. Since the
  selection event is determined entirely by $M(u_Pt\sigma + z_P) = A$,
  conditioning further on $\{ M(y) = A \}$ has the effect of
  truncating $T_P$ to $M_A(P)$, so
  \begin{equation}
    \label{eq:thm1cond}
    T_P | (A, z_P, u_P) \sim \left.\chi_r\right|_{M_A(P)}.
  \end{equation}
  Now we let $P = P_A$, and note that $T_{P_A}, z_{P_A}, u_{P_A}$
  depend on $P$ only through $A$, hence the same is true of $M_A(P)$.
  So since $P_A$ is fixed on $\{ M(y) = A \}$, the conclusion
  \eqref{eq:thm1cond} still holds with this random choice of $P$.

  This establishes \eqref{eq:tchilaw}, and
  \eqref{eq:tchisurv} follows by application of the truncated
  survival function transform, or, equivalently, truncated CDF
  transform.
\end{proof}

Since the right hand side of \eqref{eq:tchisurv} does not functionally
depend on any of the conditions involving $A$, the $p$-value result
holds unconditionally. But, in this case the meaning of the left hand
side cannot be interpreted as a test statistic for a fixed group of
variables in a specific model. We interpret the $p$-value $T\chi$
conditionally on selection so that it has the desired meaning.

This test is not optimal in the selective UMPU sense described in
\cite{fithian2014optimal}. We do not assume the linear model is
correct and we condition on $z$, this corresponds to the ``saturated
model'' in their paper. We do this for computational reasons as
described in the next section. Optimizing computation to make
increased power feasible by not conditioning on $z$ is an area of
ongoing work.

\subsection{Computing the $T\chi$ truncation interval}

Computing the support of $T\chi$ is possible due to
Lemma~\ref{lem:groupfsquad} since 
\begin{equation}
  \begin{aligned}
  \label{eq:tchiint}
  M_A &= \{ t \geq 0 : M(ut\sigma + z) = A \} \\
  &= \bigcap_{s=1}^S \bigcap_{\substack{g \in A^c_s \\ g \neq g_s}} \{
  t \geq 0 : (ut\sigma+z)^TQ^s_g(ut\sigma+z) + k^s_g \geq 0 \} \\
  &= \bigcap_{s=1}^S \bigcap_{\substack{g \in A^c_s \\ g \neq g_s}} \{
  t \geq 0 : a^s_g t^2 + b^s_g t + c^s_g \geq 0 \}
  \end{aligned}
\end{equation}
with
\begin{equation}
  \label{eq:quadcoefs}
  a^s_g := \sigma^2 u^TP_sQ^s_gP_su, \quad b^s_g := 2\sigma u^T P_sQ^s_gP_s
  z, \quad c^s_g := z^TP_sQ^s_gP_sz + k^s_g 
\end{equation}
Each set in the above intersection can be computed from the roots of
the corresponding quadratic in $t$, yielding either a single interval
or a union of two intervals. Computing $M_A$ as the intersection of
these unions of intervals can be done $O(\binom{G}{S})$. Empirically
we observe that the support $M_A$ is almost always a single interval,
a fact we hope to exploit in further work on sampling.

Since each matrix $P^s_g$ is low rank, in practice we only store the
left singular values which are sufficient for computing the above
coefficients. This improves both storage since we do not store the
full $n \times n$ projections, and computation since the coefficients
\eqref{eq:quadcoefs} require only several inner products rather than a
full matrix-vector product. 

\section{Truncated $F$ significance test}
\label{sec:TF}

We next turn to the case when $\sigma^2$ is unknown. We will see that
this example does not satisfy Definition~\ref{def:quadratic}, but the
spirit of the approach remains the same.
A selective $F$
test was first explored by \cite{gross2015internal} for affine model
selection procedures.
Following these authors, we write
\begin{equation}
\begin{aligned}
  \label{eq:TFdecomp}
  P_{\text{sub}} &:= P_{A \slash g}, \quad P_{\text{full}} := P_A\\
  R_1 &:= (I-P_{\text{sub}})y, \quad R_2 := (I-P_{\text{full}})y.
\end{aligned}
\end{equation}
and consider the $F$ statistic
\begin{equation}
  \label{eq:TF}
  T := \frac{\| R_1 \|_2^2 - \|R_2\|_2^2}{c\|R_2\|_2^2}, \text{ with
  } c := \frac{\Tr (P_{\text{full}} - P_{\text{sub}})}{\Tr (I-P_{\text{full}})} >
  0 
\end{equation}
Writing $u := R_1/r$ with $r = \|R_1\|_2$, we have the following decomposition
\begin{equation}
\begin{aligned}
  \label{eq:TFv}
  u &= u(T) = g_1(T) v_\Delta + g_2(T) v_2  \\
  \text{ with } v_\Delta &:= \frac{R_1-R_2}{\|R_1-R_2|_2^2}, \quad v_2 := \frac{R_2}{\|R_2\|_2} \\
   g_1(T) &:= r \sqrt{\frac{cT}{1+cT}}, \text{ and }
   g_2(T) := r \frac{1}{\sqrt{1+cT}}.
\end{aligned}
\end{equation}
Then $y = ru(T) + z$ for $z = P_{\text{sub}}y$.
After conditioning on $(r, z, v_\Delta, v_2)$ the only remaining
variation is in $T$ and hence $u(T)$. Without model selection, the
usual regression test for such nested models would be
\begin{equation}
  \label{eq:Fsaturatednull}
  T \sim F_{\Tr(P_A-P_{A\slash g}), \Tr(I-P_A)} \text{ under }
  H_0(A,g)  
\end{equation}

\begin{thm}
  \label{thm:tf}
  If $y \sim N(\mu, \sigma^2I)$ with $\sigma^2$ unknown, then with
  definitions \eqref{eq:TFdecomp}, \eqref{eq:TF}, \eqref{eq:TFv},
  the truncation interval $M_A := \{ t \geq 0 : M(u(t) + z) = A$,
  $d_1 := \Tr(P_A-P_{A\slash g})$, $d_2 := \Tr(I-P_A)$, under
  $H_0(A,g)$ we have 
  \begin{equation}
    \label{eq:tflaw}
    T | (z, u, A) \sim \left.F_{d_1, d_2}\right|_{M_A}
  \end{equation}
  where the vertical bar denotes truncation. Hence, with $f_{d_1, d_2}$ the pdf
  of an $F_{d_1, d_2}$ random variable
  \begin{equation}
    \label{eq:tfsurv}
    T_F := 
    \frac{\int_{M_A \cap [T, \infty]} f_{d_1, d_2}(t)dt }{\int_{M_A}
      f_{d_1, d_2}(t)dt} \sim U[0,1]
  \end{equation}
  is a $p$-value conditional on $(v_\Delta, v_2, A)$. Since the right
  hand side does not depend on these conditions, \eqref{eq:tfsurv}
  also holds marginally.
\end{thm}
\begin{proof}
  Using the fact that $(R_1 - R_2, R_2)$ are independent, the rest of
  the proof follows the same argument as Theorem \ref{thm:tchi}.
\end{proof}
Conditioning on $(z, v_\Delta, v_2)$ reduces power, but it is unclear
how to compute the truncation region without using this
strategy. As we see next, this is already non-trivial as a
one dimensional problem for quadratic selection regions.

\subsection{Computing the $T_F$ truncation interval}

Instead of a quadratically parametrized curve through the selection
region as we saw in the $T\chi$ case, we now must compute positive
level sets of functions of the form
\begin{equation}
\begin{aligned}
  \label{eq:TFquadform}
I_{Q,a,b}(t) &:=
   [z + r u(t)]^TQ[z+ru(t)] + a^T[z+ru(t)] + b \\
   &=  g_1^2 x_{11} + g_1g_2 x_{12} + g_2^2 x_{22} + g_1 x_1 + g_2x_2 + x_0
\end{aligned}
\end{equation}
where
\begin{equation}
\begin{aligned}
  \label{eq:TFquadexpanded}
g_1(t) = r\sqrt{\frac{ct}{1+ct}}, \quad g_2(t) = r\frac{1}{\sqrt{1+ct}} \\
x_{11} := v_\Delta^TQv_\Delta, \quad x_{12} := 2v_\Delta^TQv_2, \\
x_{22} := v_2^TQv_2, \quad x_1 := 2v_\Delta^TQz + a^Tv_\Delta, \\
x_2 := 2v_2^TQz + a^Tv_2, \quad x_0 := z^TQz + a^Tz + b. \\
\end{aligned}
\end{equation}
By continuity the positive level sets $\{ t \geq 0 : I_{Q,a,b}(t) \geq
0 \}$ are unions of intervals, and the selection event is an
intersection of sets of this form. We can use the same strategy as
before solving for each one and then intersecting the unions of
intervals. However, the curve is no longer quadratic, it is not clear
how many roots it may have, and its
derivative near 0 may approach $\infty$. So finding the roots is not
trivial. Our approach begins by reparametrizing with trigonometric
functions. There is an associated complex quartic polynomial and we
solve for its roots first using numerical algorithms specialized for
polynomials. A subset of these are potential roots of the original
function. This allows us to isolate the roots of the original function
and solve for them numerically in bounded intervals containing only
one root. The details are as follows. 

Since $c > 0$ and $t \geq 0$, we can make the substitution $ct =
\tan^2(\theta)$ with $0 \leq \theta < \pi/2$. Then 
\begin{equation}
g_1(\theta) = r \sin(\theta), \quad g_2(\theta) = r\cos(\theta)
\end{equation}
Using Euler's formula,
\begin{equation*}
  \begin{aligned}
    g_1(\theta)^2 &= r^2 \sin^2(\theta) = \frac{r^2}{2}\left[1 -
      \cos(2\theta)\right] = \frac{r^2}{4} \left[ 2 -
      e^{i2\theta}+e^{-i2\theta} \right] \\ 
    g_2(\theta)^2 &= r^2 \cos^2(\theta) = \frac{r^2}{2}\left[1 +
      \cos(2\theta)\right] = \frac{r^2}{4} \left[ 2 +
      e^{i2\theta}+e^{-i2\theta} \right] \\ 
    g_1(\theta)g_2(\theta) &= r^2
    \left[\frac{e^{i\theta}-e^{-i\theta}}{2i}\right]\left[\frac{e^{i\theta}+e^{-i\theta}}{2}\right]
    = -\frac{ir^2}{4}\left[e^{i2\theta} - e^{-i2\theta}\right] 
  \end{aligned}
\end{equation*}
motivates us to consider the complex function
\begin{equation*}
  \begin{aligned}
  \label{eq:TFquadcomplex}
  p(z) &:= -\frac{r^2}{4}(z^2+z^{-2}-2)x_{11} +
    -\frac{ir^2}{4}(z^2-z^{-2})x_{12} +
    \frac{r^2}{4}(z^2+z^{-2}+2)x_{22} \\ 
    & \qquad {} - \frac{ir}{2}(z-z^{-1})x_1 + \frac{r}{2}(z+z^{-1})x_2 + x_0 \\
    &= \frac{r}{2}\Big[ \frac{r}{2}(-x_{11} - i x_{12} + x_{22})z^2 +
    \frac{r}{2}(-x_{11} + i x_{12} + x_{22})z^{-2} + (-ix_1+x_2)z \\ 
    & \qquad {} + (ix_1 + x_2)z^{-1} + (rx_{11}+rx_{22} + 2x_0/r) \Big]
  \end{aligned}
\end{equation*}
which agrees with~\eqref{eq:TFquadform} when $z = e^{i\theta}$. Hence,
the zeroes of $I(t)$ coincide with the zeroes of the polynomial
$\tilde p(z) := z^2p(z)$ with $z = e^{i\theta}$ and $0 \leq \theta
\leq \pi/2$. This is the polynomial we solve numerically, as described
above. Finally, transforming the polynomial roots back to the original
domain may not be numerically stable, which is why we use a numerical
method to solve for them again. In the \texttt{selectiveInference}
software we use the R functions \texttt{polyroot} for the polynomial
and \texttt{uniroot} in the original domain. This customized approach
is highly robust and numerically stable, enabling accurate computation
of the selection region.

\section{Choosing a model with \textsc{AIC}}
\label{sec:aic}

Instead of running forward stepwise for $S$ steps with $S$ chosen
\textit{a priori}, we would like to choose when to stop adaptively.
An \textsc{AIC}-style criterion is one classical way of
accomplishing this, choosing the $s$ that minimizes
\begin{equation}
  -2 \log(L_s) + k \cdot \text{edf}_s
\end{equation}
where $L_s$ is the likelihood and $\text{edf}_s$ denotes the effective
degrees of freedom of the model at step $s$. As noted in
Section~\ref{sec:groupfs}, with $k = 2$ this is the usual
AIC, with $k = \log(n)$ it is \textsc{BIC}, and with $k = 2\log(p)$ it
is RIC.
Note that $\text{edf}$
includes whether the error variance $\sigma$ has been estimated. For
example, for a linear model with unknown $\sigma$ minimizing the above
is equivalent to minimizing
\begin{equation}
  \log(\| y - X_s \beta_s \|_2^2) + k \cdot (1+\|\beta_s\|_0)/n
\end{equation}
where $\|\cdot\|_0$ counts the number of nonzero entries. Instead of
taking the approach which minimizes this over $1 \leq s \leq S$, we
adopt an early-stopping rule which picks the $s$ after which the
\textsc{AIC} criterion increases $s_+$ times in a row, for some $s_+
\geq 1$. Readers familiar with the R programming language might
recognize this as the default behavior of \texttt{step} when $s_+ =
1$.

Fortunately, since the stopping rule depends only on successive values
of the quadratic objective which we are already tracking, we can
condition on the event $\{ \hat s = s \}$ by appending some additional
quadratic inequalities encoding the successive comparisons. The
results concerning computing $T\chi$ or $T_F$ remain intact, with the
only complication being the addition of a few more inequalities in the
computation of $M_A$.

\section{Theory and power}
In this section we describe the behavior of the $T\chi$ statistic
under some simplifying assumptions and evaluate power against some
alternatives.

\begin{defn}
  By groupwise orthogonality or orthogonal groups we mean
  that $X_g^TX_h = 0_{p_g  \times p_h}$ for all $g, h$, and the
  columns of $X$ satisfy $\| X_j \|_2 = 1$ for all $j$.
\end{defn}

\begin{prop}
  \label{prop:order}
  Let $T_s$ denote the observed $T\chi$ statistic for the group that
  entered the model at step $s$. With orthogonal groups of equal size
  these are ordered $T_{1} > T_{2} > \cdots > T_{S}$.
  Further, the truncation regions are given exactly
  by this ordering along with $T_1 < \infty$, $T_S > T_{S+1}$ where
  $T_{S+1}$ is the $T\chi$ statistic for the next variable that would
  enter the model, understood to be 0 if there are no remaining variables.
\end{prop}
\begin{proof}
  For simplicity we assume $S = G$ so that $T_{S+1} = 0$.
  For each $g$ let $X_g = U_gD_gV_g^T$ be a singular value decomposition.
  Hence $U_g$ is a matrix constructed of orthonormal columns forming a
  basis for the column space of $X_g$. The first group to enter
  $g_1$ satisfies
  \begin{equation}
  g_1 = \argmax_g \| U_g^T y \|
  \end{equation}
  Note that $r_1 = y - U_{g_1}U_{g_1}^Ty$ is the residual after the
  first step, and that $U_g - U_{g_1}U_{g_1}^TU_g = U_g$ for all $g
  \neq g_1$ by orthogonality. The second group $g_2$ to enter
  satisfies
  \begin{equation}
  g_2 = \argmax_{g \neq g_1} \| U_g^T r_1 \| = \argmax_{g \neq g_1} \| U_g^T y \|
  \end{equation}
  because $U_g^TU_{g_1} = 0$ implies $U_g^T r_1 = U_g^T y$. By
  induction it follows that the $T_i = \| U_{g_i}^T y \|$ are
  ordered.

  Since we have assumed all groups have equal size,
  the truncation region for $T_s$ is determined by the intersection of
  the inequalities
  \begin{equation}
  (\eta_{s,i} t + z_{s,i})^T(U_{g_i}U_{g_i}^T - U_{g_j}U_{g_j}^T)(\eta_{s,i}t + z_{s,i})
  \geq 0 \quad \forall i,j \text{ s.t. } i \leq j \leq G.
  \end{equation}
  By orthogonality, $\eta_{s,i} = \eta_s \propto U_{g_s}U_{g_s}^Ty$
  and $z_{s,i} = z_s = y - U_{g_s}U_{g_s}^Ty$ for all $i$.
  Hence if $s \notin \{ i, j \}$ the inequality is satisfied for
  all $t$. If $s = i$, after expanding the inequalities become $t^2 -
  T^2_j \geq 0$ for all $j$, and by the ordering this reduces to $t
  \geq T_{g_{s+1}}$. The case $s = j$ is similar and yields one upper
  bound $T_{g_{s-1}} \geq t$ (implicitly we have defined $T_{g_0} =
  \infty$).
\end{proof}
Proposition~\ref{prop:order} allows us to analyze the power in the
orthogonal groups case by evaluating $\chi$ survival functions with
upper and lower limits given simply by the neighboring $T\chi$
statistics. To apply this we still need to deal with the fact that the
order of the noncentrality parameters may not correspond to the order
of groups entering the model, and it is also possible for null
statistics to be interspersed with the non-nulls.

\begin{figure}[h!]
  \centering
  \includegraphics[width=\textwidth]{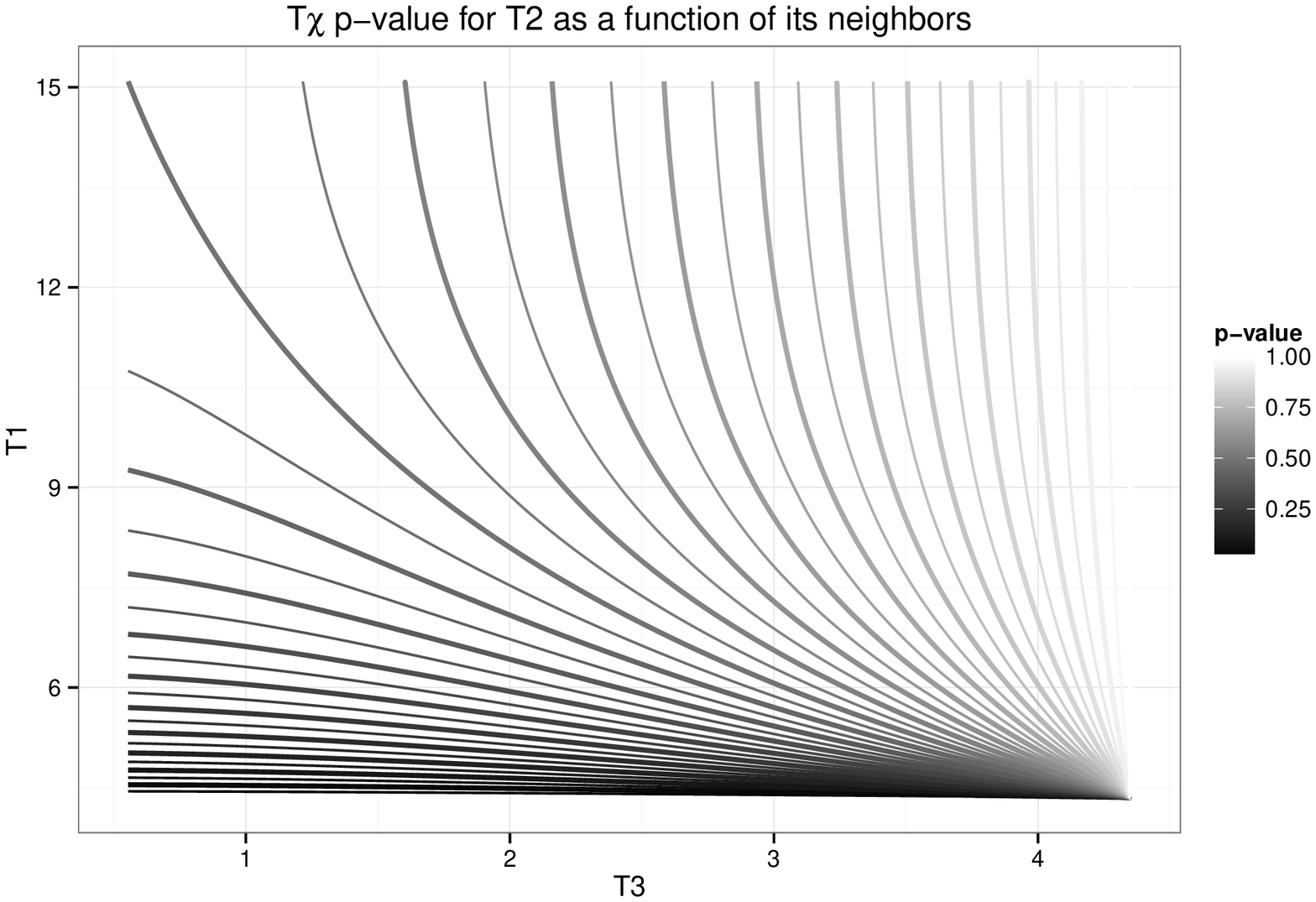}
  \caption{Contours of $T\chi$ $p$-value for $T_2$ as a function of
    its neighbors. Here $T_2 \approx 4.35$ (corresponding to the 50\%
    quantile of a $\chi^2_5$) is fixed and the upper and lower
    truncation limits vary.}
  \label{fig:powercontour}
\end{figure}

Proposition~\ref{prop:order} also
yields some heuristic understanding. For example, we see
explicitly the non-independence of the $T\chi$ statistics, but also
note that in this case the dependency is only on neighboring
statistics. Further, we see how the power for a given test depends on
the values of the neighboring statistics. An example with groups of
size 5 is plotted in Figure~\ref{fig:powercontour}. This example
considers a fixed value of $T_2$ and plots contours of the $p$-value
for $T_2$ as a function of the neighboring statistics. It is apparent
that when $T_3 \to T_2$ (right side of plot),
the $p$-value for $T_2$ will be large irrespective of $T_1$.

This motivates us to consider when a given nonnull $T_i$ is likely
to have a close lower limit, since that scenario results in low
power. Since a non-central $\chi^2_k(\lambda)$ with non-centrality
parameter $\lambda$ has mean $k+\lambda$ and variance $2k+4\lambda$,
we expect that the nonnull $T_i$ will be larger and have greater
spread than the $T_i$ corresponding to null variables. Hence, if the
true signal is sparse then the risk
of having a nearby lower limit is mostly due to the bulk of null
statistics. By upper bounding the largest null statistic, we get an
idea of both when forward stepwise will select the nonnull variables
first and when the worst case for power can be excluded.  
Using Lemma 1 of \cite{laurent2000adaptive}, for $\epsilon > 0$ if we
define
\begin{equation}
 x = -\log(1-(1-\epsilon)^{1/G})
\end{equation}
then if all $T_i$ are null (central) we have
\begin{equation}
 \mathbb P\left(\max_{1 \leq i \leq G} T_i^2 > k + 2\sqrt{kx} + 2x\right) < \epsilon.
\end{equation}
Table~\ref{tab:screenbounds} gives values of the upper bound
$k + 2\sqrt{kx} + 2x$ as a 
function of $G$ and $k$, where there are $G$ null groups of equal size $k$
and error variance $\sigma^2 = 1$. For example, with $G = 50$ groups
of size $k = 2$ the null $\chi^2_k$ statistics will all be below 27.28 with
99\% probability and 21.35 with 90\% probability.

    \begin{table}[ht]
      \caption{High probability upper bounds for central $\chi^2_k$
        statistics corresponding noise variables. The left panel gives
        99\% bounds and the right panel 90\%. This assumes $G$ noise
        groups of equal size $k$, groupwise orthogonality, and error
        variance $\sigma^2 = 1$.} 
      \label{tab:screenbounds}
      \centering
      \begin{tabular}{r|rrrr}
        & \multicolumn{4}{c}{$k$} \\
        \hline
        $G$ & 2 & 5 & 10 & 50 \\
        \hline
        10 & 23.24 & 30.56 & 40.42 & 100.96 \\
        20 & 24.99 & 32.52 & 42.62 & 104.17 \\
        50 & 27.28 & 35.07 & 45.48 & 108.29 \\
        100 & 28.99 & 36.98 & 47.60 & 111.32 \\
        1000 & 34.61 & 43.19 & 54.47 & 120.99 \\
        \hline
      \end{tabular}
      \quad
      \begin{tabular}{r|rrrr}
        & \multicolumn{4}{c}{$k$} \\
        \hline
        $G$ & 2 & 5 & 10 & 50 \\
        \hline
        10 & 17.16 & 23.66 & 32.62 & 89.31 \\
        20 & 18.98 & 25.74 & 34.99 & 92.90 \\
        50 & 21.35 & 28.43 & 38.03 & 97.44 \\
        100 & 23.12 & 30.42 & 40.27 & 100.74 \\
        1000 & 28.88 & 36.85 & 47.46 & 111.11 \\
        \hline
      \end{tabular}
    \end{table}

\begin{figure}[!tbp]
  \centering
  \subfloat[Probability of a single nonnull exceeding
  the 90\% theoretical bounds in Table~\ref{tab:screenbounds}.]
  {\includegraphics[width=\textwidth]{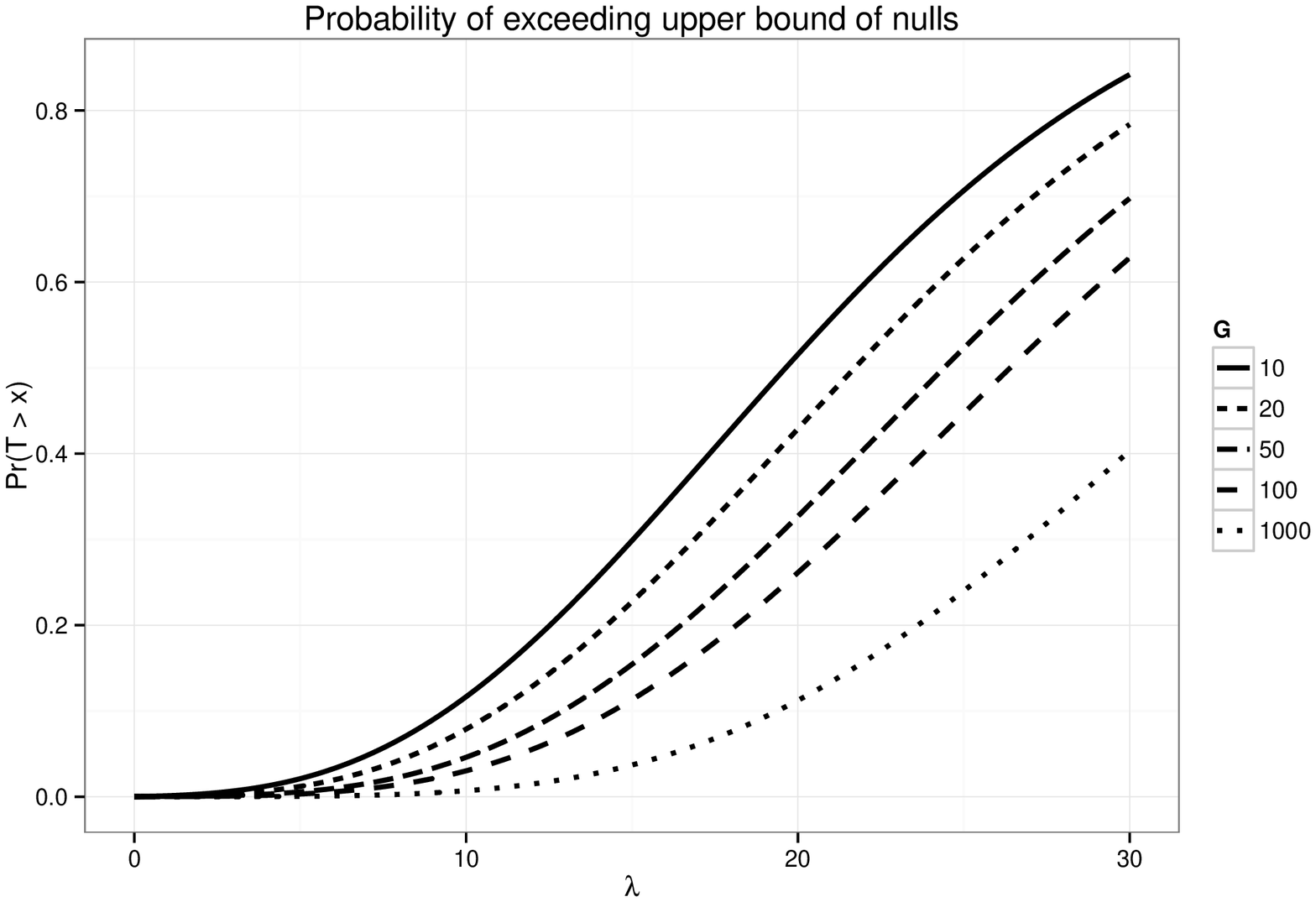}\label{fig:f1}}
\\
  \subfloat[Smoothed estimates of power functions for 1-sparse
  alternative and $G$ null variables.]
  {\includegraphics[width=\textwidth]{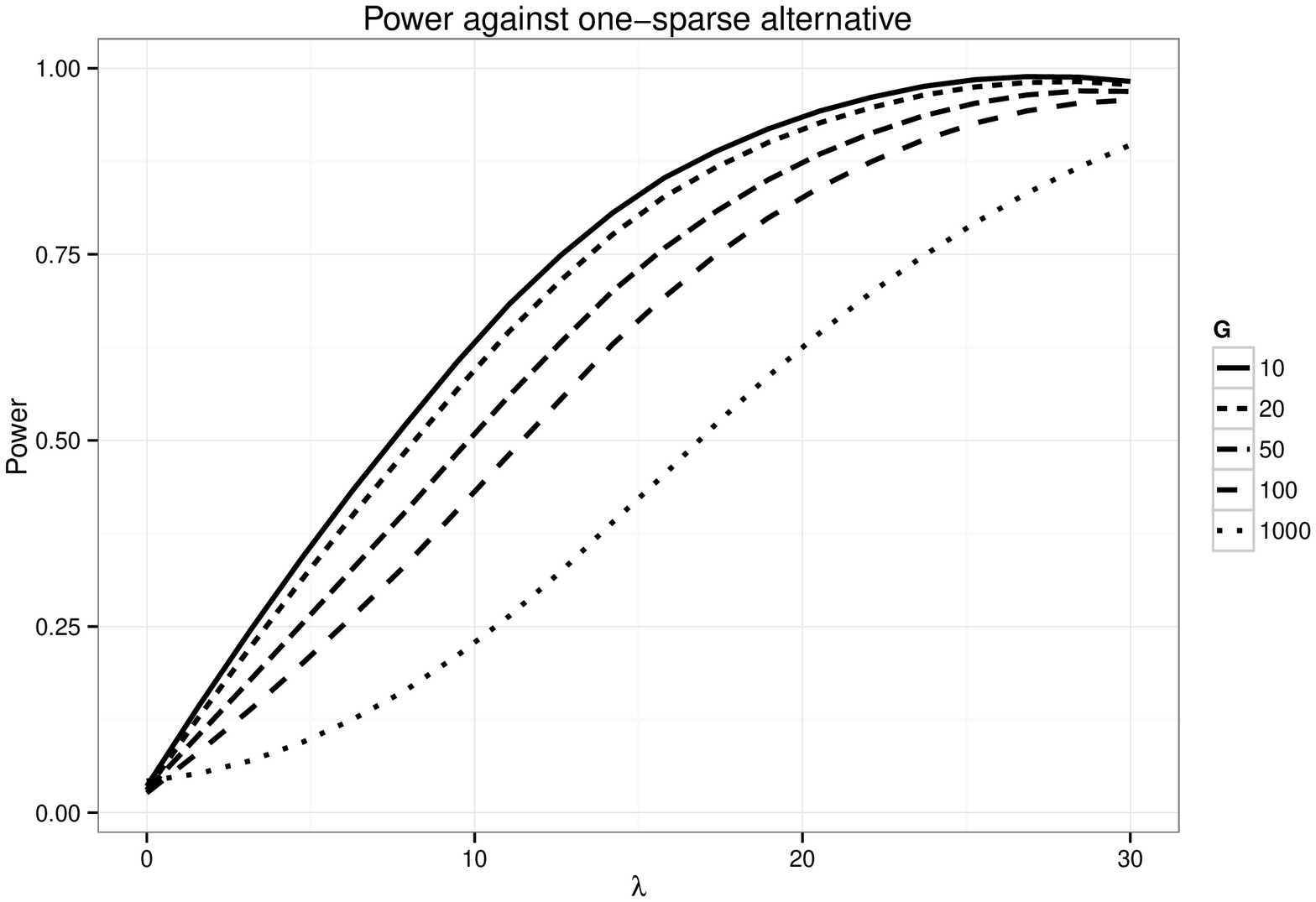}\label{fig:f2}}
  \caption{Theoretical and empirical curves related to power against a
    one-sparse non-central $\chi^2_5$ with non-centrality parameter
    $\lambda$.}
  \label{fig:powercurves}
\end{figure}

The power curves plotted in Fig~\ref{fig:powercurves} show this
theoretical bound is likely more conservative than necessary.
In fact, let us simplify a bit more by assuming orthogonal groups of
size 1 (single columns), and a 1-sparse alternative with magnitude
$(1+\delta)\sqrt{2\log(p)}$ for $\delta > 0$. Then the standard tail
bounds for Gaussian random variables imply that the $T\chi$ test is
asymptotically equivalent to Bonferroni and hence asymptotically
optimal for this alternative. This is discussed in
\cite{loftus2014significance} and a little more generally in
\cite{taylor2013tests}. The 1-sparse case with orthogonal groups of
fixed, equal size follows similarly using tail bounds for $\chi^2$
random variables.

To translate between the non-centrality parameter $\lambda$ and a linear
model coefficient, consider that when $y = X_g \beta_g + \epsilon$ and
the expected column norms of $X_g$ scale like $\sqrt{n}$, then the
non-centrality parameter for $T_g$ will be roughly equal to $\sqrt{n}
\| \beta_g\|_2^2$ (under groupwise orthogonality).

Finally, when groups are not orthogonal, the greedy nature of forward stepwise
will generally result in less power since part of the variation in
$y$ due to $X_g$ may be regressed out at previous steps.  

\section{Simulations}
\label{sec:sim}

\begin{figure}[h!]
  \centering
  \includegraphics[width=\textwidth]{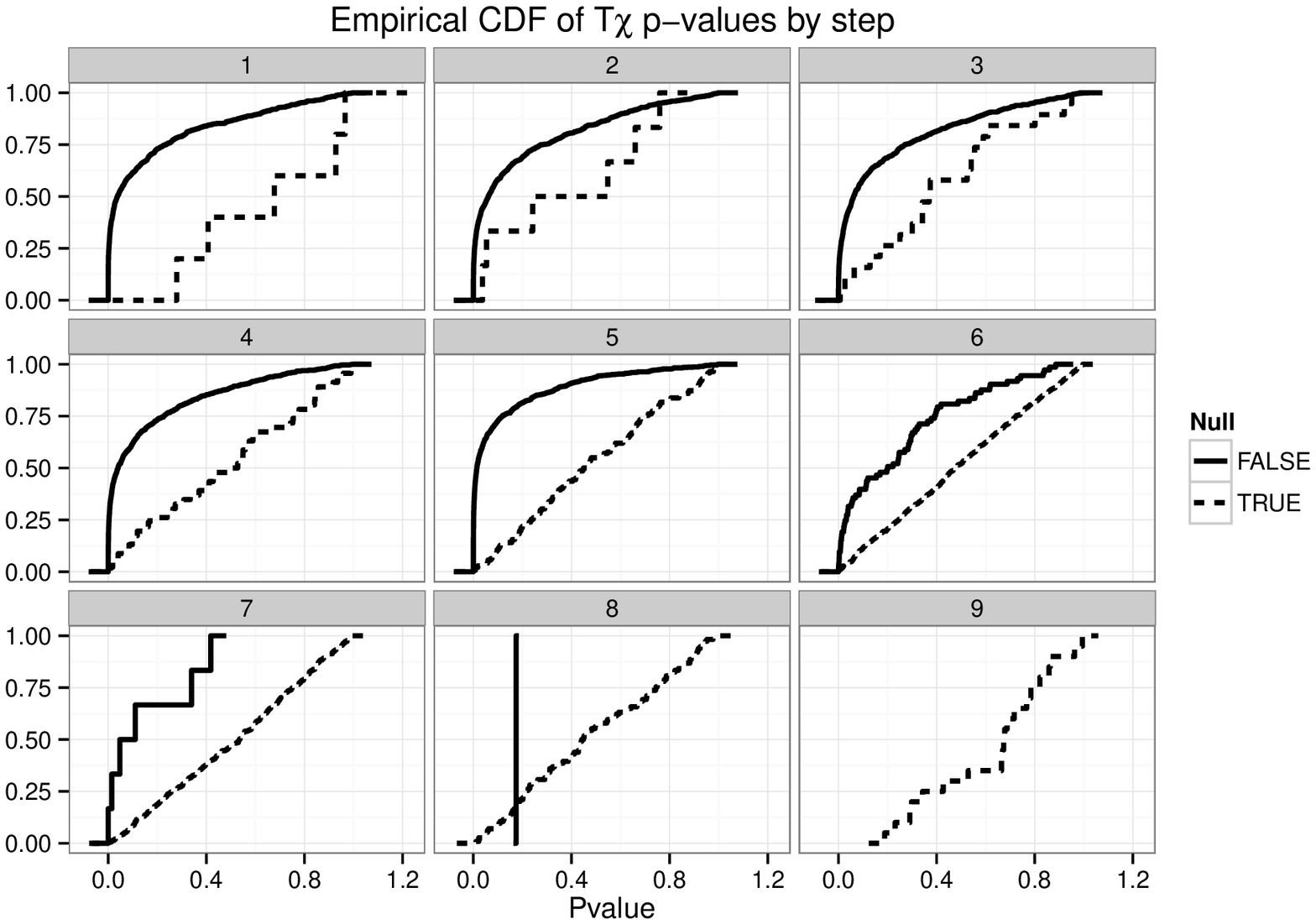}
  \caption{Empirical CDFs of $T\chi$ $p$-values plotted by
    step. Models chosen by BIC, with model size distribution described
    in Table~\ref{tab:bicsizes}. Design matrix had $n = 100$
    observations and $p = 100$ variables in $G = 50$ groups of size
    2. The true sparsity is 5, nonzero entries of $\beta$ equal
    $\pm 2\sqrt{\log(G)/n} \approx \pm 0.4$. Dotted lines correspond
    to null variables and the solid, off-diagonal lines to signal
    variables. }
  \label{fig:ecdfbystepbic}
\end{figure}

To evaluate the $T\chi$ test when groupwise orthogonality does not
hold, we conduct a simulation experiment with correlated Gaussian
design matrices. In each of 1000 realizations forward stepwise is run
with number of steps chosen by BIC. Table~\ref{tab:bicsizes}
summarizes the resulting model sizes and how many of the 5 nonnull
variables were included.

\begin{table}[ht]
\caption{Size and composition of models chosen by BIC in the
  simulation described in Fig~\ref{fig:ecdfbystepbic}.} 
\label{tab:bicsizes}
\centering
\begin{tabular}{r|rrrrrr}
  \hline
Model size & 4 & 5 & 6 & 7 & 8 & 9 \\ 
  \# Occurrences & 4 & 42 & 511 & 328 & 95 & 20 \\ 
   \hline
\end{tabular}
\begin{tabular}{r|rrr}
\# True signals included & 3 & 4 & 5 \\ 
\# Occurrences &  18 & 106 & 876 \\ 
   \hline
\end{tabular}
\end{table}

Fig \ref{fig:ecdfbystepbic} shows the empirical distributions of both
null and nonnull $T\chi$ $p$-values plotted by the step the
corresponding variable was added. Table~\ref{tab:power} shows the
power of rejecting at the $\alpha = 0.05$ level also by step. Note
that it was rare for nonnull variables to enter at later steps, so the
observed power of 0 at step 8 is not too surprising.

\begin{table}[ht]
\caption{Empirical power for nonnull variables in the simulation
  described in Fig~\ref{fig:ecdfbystepbic}.}
\label{tab:power}
\centering
\begin{tabular}{r|llllllll}
  \hline
  step & 1 & 2 & 3 & 4 & 5 & 6 & 7 & 8 \\ 
  Power & 0.532 & 0.470 & 0.454 & 0.521 & 0.641 & 0.315 & 0.500 & 0.000 \\ 
   \hline
\end{tabular}
\end{table}

\bibliographystyle{imsart-nameyear.bst}
\bibliography{biblio}

\end{document}